\documentclass[10pt, onecolumn]{IEEEtran}
\usepackage{graphicx}

\usepackage{amsmath}
\usepackage{amssymb}
\usepackage{amsfonts}
\usepackage{epsfig}

\usepackage{amsmath,mathtools}    
\usepackage{graphicx}   
\usepackage{color}      
\usepackage{hyperref}
\usepackage{comment}
\usepackage{subfigure}
\usepackage{times}
\usepackage{hyphenat}
\usepackage{epsfig}
\usepackage{latexsym}
\usepackage{bbold,bbm}
\usepackage{setspace}
\usepackage{amssymb,amsmath}
\usepackage{mathrsfs}
\usepackage{amsgen,amsfonts,amsbsy,amsthm}
\usepackage{algorithmic,algorithm}
\usepackage{array}
\usepackage{booktabs}
\usepackage{amsfonts}
\usepackage{amssymb}
\usepackage{multirow}
\usepackage{longtable}
\usepackage{epstopdf}
\usepackage{wrapfig}
\usepackage[font=small]{caption}

\newtheorem{defn}{Definition}
\newtheorem{thm}{Theorem}

\newtheorem{cor}[thm]{Corollary}

\newtheorem{const}{Construction}
\newtheorem{claim}{Claim}

\newcommand{\bit}{\begin{itemize}}
\newcommand{\eit}{\end{itemize}}
\newcommand{\bcor}{\begin{cor}}
\newcommand{\ecor}{\end{cor}}
\newcommand{\beq}{\begin{equation}}
\newcommand{\eeq}{\end{equation}}
\newcommand{\beqn}{\begin{equation*}}
\newcommand{\eeqn}{\end{equation*}}
\newcommand{\bea}{\begin{eqnarray}}
\newcommand{\eea}{\end{eqnarray}}
\newcommand{\bean}{\begin{eqnarray*}}
\newcommand{\eean}{\end{eqnarray*}}
\newcommand{\ben}{\begin{enumerate}}
\newcommand{\een}{\end{enumerate}}
\newcommand{\bdefn}{\begin{defn}}

\renewcommand\footnotemark{}

\begin{document}
\sloppy
\title{Binary Codes with Locality for Four Erasures}

\author{
\IEEEauthorblockN{S. B. Balaji, K. P. Prasanth and P. Vijay Kumar, \it{Fellow}, \it{IEEE}}

\IEEEauthorblockA{Department of Electrical Communication Engineering, Indian Institute of Science, Bangalore.  \\ Email: balaji.profess@gmail.com, prasanthkp231@gmail.com, pvk1729@gmail.com} 

\thanks{P. Vijay Kumar is also an Adjunct Research Professor at the University of Southern California.  This research is supported in part by the National Science Foundation under Grant 1421848 and in part by an India-Israel UGC-ISF joint research program grant.} 
\thanks{S. B. Balaji would like to acknowledge the support of TCS research scholarship program.}
}
\maketitle

\begin{abstract}
In this paper, codes with locality for four erasures are considered. An upper bound on the rate of codes with locality with sequential recovery from four erasures is derived. The rate bound derived here is field independent. An optimal construction for binary codes meeting this rate bound is also provided. The construction is based on regular graphs of girth $6$ and employs the sequential approach of locally recovering from multiple erasures. An extension of this construction that generates codes which can sequentially recover from five erasures is also presented.
\end{abstract}



\section{Introduction}
An $[n,k]$ code is said to have locality $r$ if each of the $n$ code symbols of $\mathcal{C}$ can be recovered by accessing at most $r$ other code symbols. Equivalently, there exist $n$ codewords ${h_1 \cdots h_n}$ in the dual code $\mathcal{C}^\perp$ such that $c_i \in \text{supp}(h_i)$ and $|\text{supp}(h_i)| \leq r+1$ for $1 \leq i \leq n$ where $c_i$ denote the $i^\text{th}$ code symbol of $\mathcal{C}$ and $\text{supp}(h_i)$ denote the support of the codeword $h_i$.
\paragraph{Codes with Sequential Recovery}
An $[n,k]$ code is defined as a code with sequential recovery \cite{BalPraKum} from $t$ erasures having locality $r$ if for any set of $s \leq t$ erased symbols, $\{c_{\sigma_1},...,c_{\sigma_s} \}$, there is an arrangement of these $s$ symbols (say) $\{c_{\sigma_{i_1}},...,c_{\sigma_{i_s}}\}$ such that there are $s$ codewords $\{h_1,...,h_s\}$ in the dual of the code, each of Hamming weight $\leq r+1$, with $\sigma_{i_j} \in \text{supp}(h_j)$ and $\text{supp}(h_j) \cap \{\sigma_{i_{j+1}},\cdots,\sigma_{i_s}\} = \emptyset$, $\forall 1 \leq j \leq s$.  The parameter $r$ is the locality parameter and we will formally refer to this class of codes as $(n,k,r,t)_{\text{seq}}$ codes. When the parameters $(n,k,r,t)$ are clear from the context, we will simply refer to a code in this class as a code with sequential recovery.
\subsection{Background}
In \cite{GopHuaSimYek} P. Gopalan et al. introduced the concept of codes with locality (see also \cite{PapDim,OggDat}), where an erased code symbol is recovered by accessing a small subset of other code symbols. The size of this subset denoted by $r$ is typically much smaller than the dimension of the code, making the repair process more efficient compared to MDS codes. The authors of \cite{GopHuaSimYek} considered codes that can locally recover from single erasures (see also \cite{GopHuaSimYek,HuaChenLi,KamPraLalKum,TamBar_Optimal_LRC})

The sequential approach introduced by Prakash et al. \cite{PraLalKum} is one of the many approaches to locally recover from multiple erasures. Codes employing this approach have been shown to be better in terms of rate and minimum distance (see \cite{PraLalKum,RawMazVis,SonYue_Binary_local_repair,BalPraKum}). The authors of \cite{PraLalKum} considered codes that can sequentially recover from two erasures (see also \cite{SonYue_3_Erasure}). Codes with sequential recovery from three erasures can be found discussed in \cite{SonYue_3_Erasure,BalPraKum}.			 

Alternate approaches for local recovery from multiple erasures can be found in \cite{KamPraLalKum,SonDauYueLi,ZhaWanGe,HuaYaaUchSie,TamBarFro,KimNamSong,SheFuGua,WanZha_Combinatorial_Repair_locality,TamBar_Optimal_LRC,JuaHolOgg,RawPapDimVis_arxiv,WangZhang_multiple_erasure}.
\subsection{Our Contributions} 
In this paper, binary codes with locality for four erasures are considered. An upper bound on the rate of $(n,k,r,4)_{seq}$ code is derived. A construction for codes achieving this rate bound with equality using the sequential approach is also provided. The construction uses regular bi-partite graphs of girth $6$. Finally, we show that this construction can be easily modified to generate codes with sequential recovery for five erasures.
\section{Upper Bound on Rate}
The following theorem gives an upper bound on the rate of $(n,k,r,4)_\text{seq}$ codes.
\begin{thm} \label{thm:rate_4}
The rate of an $(n,k,r,4)_\text{seq}$ code over a field $\mathbb{F}_q$ satisfies the following upper bound:
\bean
\frac{k}{n} \leq \frac{r^2}{r^2+2r+2}.
\eean
\end{thm}
\begin{proof}
Let $\mathcal{C}$ be an $(n,k,r,4)_\text{seq}$ code. Let,
\begin{eqnarray} 
\mathcal{B}_{0} & = & \text{span}\left({\bf c} \in \mathcal{C}^{\perp}, |\text{supp}({\bf c})| \leq r+1 \right).
\end{eqnarray} 
i.e., $\mathcal{B}_{0}$ is the span of all local parity-checks (those parity-checks having Hamming weight $\leq(r+1)$). Let $m$ denote the dimension of the subcode $\mathcal{B}_{0}$.
We have 
\beq
n-k \geq m. \label{eq:n_k_bnd}
\eeq
Choose $m$ linearly independent vectors from the set $\{{\bf c} \in \mathcal{C}^{\perp}, |\text{supp}({\bf c})| \leq r+1\}$, and form an $(m \times n)$ matrix $H'$ with these $m$ vectors as its rows. Let $s_1$ denote the number of columns of $H'$ having a Hamming weight of $1$ and $s_2$ denote the number of columns of $H'$ having a Hamming weight of $2$ . Permute the rows and columns of the matrix $H'$ to get a matrix $H$ as shown in \eqref{eq:4_erasure_matrix}.	Note that this permutation does not affect the rank of the matrix or the Hamming weights of its rows/columns.
\bea
H & = &  \left[ 
\begin{array}{c|c|c|c}
D_{s_1}  & A & 0 & \multirow{2}{*}{D}\\
\cline{1-3}
0 & B & C &			
\end{array}
\right] \label{eq:4_erasure_matrix}
\eea 
The sub-matrix $\left[\frac{D_{s_1}}{0}\right]$ is an $(m \times s_1)$ matrix comprising the $s_1$ columns of $H$ having weight $1$. Note that out of the $s_1$ columns having weight $1$, no two columns can have a non-zero entry in the same row, as those two columns would form a set of two linearly dependent columns preventing local recovery from $2$ erasures. Therefore, WLOG we can assume that $D_{s_1}$ is an $(s_1 \times s_1)$ matrix with $s_1$  non-zero entries on its diagonal and zeros elsewhere.

The sub matrix $\left[\frac{A}{B}\right]$ comprises those columns of $H$ having weight $2$ with the first non zero element in the first $s_1$ rows and the second non zero element in the next $(m-s_1)$ rows. i.e., $A$ and $B$ are matrices having columns of Hamming weight one each. Note that among the columns having weight $2$, no column can have both its non zero entries in the first $s_1$ rows, otherwise the column with both non zero entries in the first $s_1$ rows along with a certain chosen set of $2$ columns of $D_{s_1}$ will form a set of three linearly dependent vectors, preventing local recovery from $3$ erasures. Hence, the remaining columns of $H$ which have a Hamming weight of $2$ form the sub matrix $\left[\frac{0}{C}\right]$. Each column of $C$ has Hamming weight $2$. The sub-matrix $D$ contains all columns of $H$ having Hamming weight three or more.

Let $s_{21}$ denote the number of columns of $\left[\frac{A}{B}\right]$ and $s_{22}$ denote the number of columns of $\left[\frac{0}{C}\right]$. Therefore, $s_{21}+s_{22}=s_2$.

Consider the matrix $B$. Assume that there is a row in $B$ with more than one non zero element. Consider the two columns where this row has its two non-zero elements. It is straightforward to see that these two columns, along with a certain chosen set of two columns from $D_{s_1}$ form a set of $4$ linearly dependent columns, making local recovery from $4$ erasures impossible. Hence each row in $B$ can have at most one non zero element. 
It follows that  (counting the number of non-zero entries in the rows and columns of $B$)
\bean
s_{21} & \leq & m-s_1.\\
\text{Assume that } \hspace{0.5cm} s_{21} & = & m-s_1-p, \hspace{0.5cm} \text{where } 0 \leq p \leq m-s_1.
\eean
Each of the $(m-s_1-p)$ (or $s_{21}$) columns of $A$ has only one non-zero element. Each of the $s_1$ rows of $A$ can have at most $r$ non zero elements. Counting the number of non-zero entries in the rows and columns of $A$ :
\bea
\therefore s_1 & \geq & \frac{m-s_1-p}{r} \notag \\
s_1 & \geq & \frac{m-p}{r+1}\label{eq:s1_lb}
\eea
The number of non zero elements in the sub-matrix $[B|C]$ is upper bounded by $(m-s_1)(r+1)$. The sub-matrix $B$ consists of $m-s_1-p$ non-zero elements (one in each column). Therefore the number of non zero elements in $C$ is upper bounded by $(m-s_1)(r+1)-(m-s_1-p)$. Each column of $C$ has $2$ non zero elements. Thus we have  (counting the number of non-zero entries in the rows and columns of $C$):
\bea 
s_{22} & \leq & \frac{(m-s_1)(r+1)-(m-s_1-p)}{2}  = \frac{(m-s_1)r+p}{2} \notag\\
\therefore \ s_2 & = &s_{21} + s_{22}  \leq  (m-s_1-p) + \frac{(m-s_1)r+p}{2} \label{eq:s2_ub}
\eea
Consider the matrix $H$. Each of the $m$ rows of $H$ can have at most $(r+1)$ non-zero values. Thus we have (counting the number of non-zero entries in the rows and columns of $H$):
\bean
s_1+2s_2+3(n-s_1-s_2) & \leq & m(r+1)\\
3n-2s_1-s_2 & \leq & m(r+1)
\eean
Substituting the upper bound on $s_2$ in \eqref{eq:s2_ub}, in the above inequality gives:
\beqn
3n+s_1\left(\frac{r}{2}-1\right)+\frac{p}{2} \leq m\left(\frac{3r}{2}+2 \right).
\eeqn
Substituting the lower bound on $s_1$ in \eqref{eq:s1_lb}, in the above inequality gives:
\bean
3n+\frac{m-p}{r+1}\left(\frac{r}{2}-1\right)+\frac{p}{2} & \leq & m\left(\frac{3r}{2}+2 \right),\\
3n & \leq & m\left(\frac{3r}{2}+2-\frac{r-2}{2(r+1)}\right) +p\left(\frac{r-2}{2(r+1)}-\frac{1}{2}\right).
\eean
The coefficient of $p$ in the above inequality is negative for $r \geq 1$. 
\bea
\therefore 3n \leq m\left(\frac{3r}{2}+2-\frac{r-2}{2(r+1)}\right) \label{eq:3n_bnd}
\eea
From \eqref{eq:n_k_bnd} and \eqref{eq:3n_bnd}, we have:
\bean
\frac{k}{n} \leq 1-\frac{m}{n} \leq \frac{r^2}{r^2+2r+2}
\eean
\end{proof}
	
\section{Optimal Construction}
In this section we provide a construction that generates $(n,k,r,4)_\text{seq}$ codes that achieve the upper bound on rate given in Theorem~\ref{thm:rate_4} with equality.
\begin{const} \label{const:four_reg}
Consider an $r$-regular bipartite graph on $2L$ nodes, having girth at least $6$. The edges of this bipartite graph represent information symbols of the code. A node in the bi-paritite graph represent a code symbol which is the parity of the $r$ information symbols corresponding to the $r$ edges connected to it. Consider $r$ copies of this graph. Let $n_l^{(i)}$ represent the $l^\text{th}$ node of the $i^\text{th}$ bipartite graph, $1 \leq l \leq 2L$, $1\leq i \leq r$. The nodes in one copy of the bipartite graph are labeled as follows. Nodes labeled $n_1^{(i)} \cdots n_L^{(i)}$ are the nodes appearing on one side, say the top, of the $i^\text{th}$ bipartite graph and nodes $n_{L+1}^{(i)} \cdots n_{2L}^{(i)}$  appear on the other side, (i.e., on the bottom). Consider the set $\{n_1^{(1)}, n_1^{(2)}, \cdots n_1^{(r)}\}$, the set of $r$ nodes appearing in the first position in the $r$ bipartite graphs. Nodes in this set are connected to a new node $N_1$ added to the graph. These connections are represented by using dotted lines as the corresponding edges do not represent code symbols. Node $N_1$ represents a new code symbol that holds the parity of the $r$ parity symbols represented  by the $r$ nodes connected to $N_1$. 

A total of $(2L)$ new nodes $\{N_1 \cdots N_{2L}\}$ are added to the graph in a similar manner. Node $N_j$ represents the code symbol that stores the parity of the $r$ symbols represented by the nodes $\{n_j^{(1)}, n_j^{(2)}, \cdots n_j^{(r)}\}$, $1 \leq j \leq 2L$, which are connected to $N_j$.

Altogether the graph has $(Lr^2)$ solid edges and $(2Lr+2L)$ nodes. Thus the code has $(Lr^2)$ information symbols and $(2Lr+2L)$ parity symbols. Hence the code has a rate of $\frac{r^2}{r^2+2r+2}$, which matches the upper bound in Theorem~\ref{thm:rate_4}. Hence the construction is rate-optimal for $t=4$.
\end{const}

An example code for the case $r=3,L=7$ is provided in Fig.~\ref{fig:four_regular}.
\begin{figure}[h!]
\centering
\includegraphics[width=6.5in]{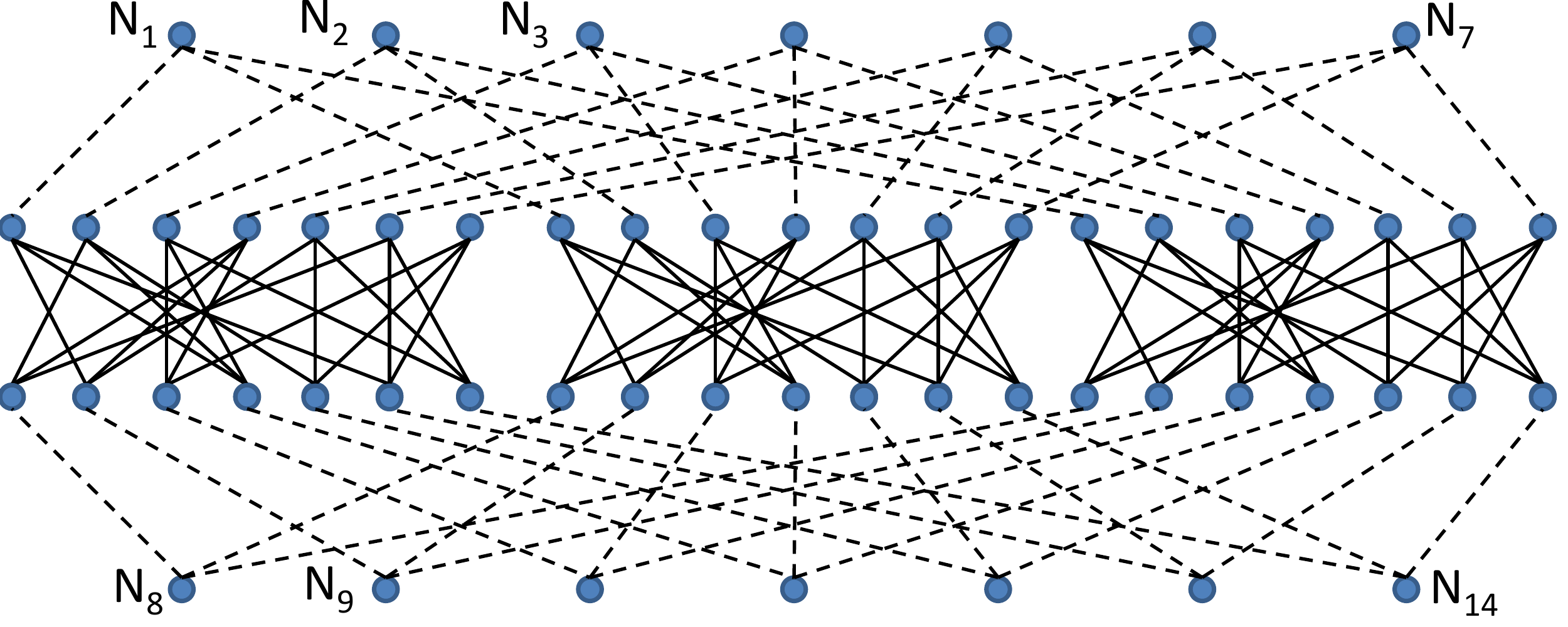}
\caption[Example bipartite graph code for $4$ erasures]{An example code obtained from Construction~\ref{const:four_reg} for $r=3,L=7$. Solid edges represent information symbols. Nodes represent parity symbols.}
\label{fig:four_regular}
\end{figure}
\begin{claim}
A code generated by Construction~\ref{const:four_reg} is an $(n,k,r,4)_\text{seq}$ code.
\end{claim}
\begin{proof}
Assume that $x$ number of information symbols are erased and $4-x$ parity symbols are erased. We consider the cases $x=0,x=1,x=2,x=3$ separately and prove that sequential recovery is possible in each case. Throughout the proof we use the terms 'nodes' (or 'edges') and 'code symbols' interchangeably. Hence, whenever it is mentioned that a node $n_l^{(i)}$ is erased, it means that the corresponding parity symbol is erased. Similarly, an edge is erased means that the corresponding information symbol is erased.
\ben
\item[(I)] \emph{$x=0 :$} Each node $n_l^{(i)}, 1 \leq l \leq 2L, \ 1 \leq i \leq r$ in the bipartite graph can be recovered locally from the edges (note that none of the edges are erased in this case). Once the erasures among nodes $n_l^{(i)}, 1 \leq l \leq 2L, \ 1 \leq i \leq r$ are recovered, the erasures among the nodes $N_1 \cdots N_{2L}$ can be recovered locally using nodes $n_l^{(i)}, 1 \leq l \leq 2L, \ 1 \leq i \leq r$. Hence, sequential recovery is possible in this case.
\item[(II)] \emph{$x=1 :$} Let $n_l^{(i)}$ and $n_{l'}^{(i)}$ be the two nodes appearing on the two ends of the one edge that is erased. Then, the erased edge can be recovered using two parity checks $P_l$ and $P_{l'}$ where, $P_l$ is the parity-check involving node $n_l^{(i)}$ and the edges connected to $n_{l}^{(i)}$, and $P_{l'}$ is the parity-check involving node $n_{l'}^{(i)}$ and the edges connected to $n_{l'}^{(i)}$. No more information symbol (i.e., edge) erasures are possible in this case. Hence the erased information symbol can be recovered locally unless both $n_j^{(l)}$ and $n_k^{(l)}$ are erased. 

Consider the case when both $n_l^{(i)}$ and $n_{l'}^{(i)}$ are erased, so that the erased information symbol cannot be recovered using $P_l$ and $P_{l'}$. Node $n_l^{(i)}$ can be recovered using the parity-check involving node $N_l$ and the $r$ nodes connected to it. Similarly, $n_{l'}^{(i)}$ can be recovered using the parity-check involving node $N_{l'}$ and the $r$ nodes connected to it. Also, note that these two parity-checks have disjoint support sets. Hence, irrespective of the fourth erasure, either $n_l^{(i)}$ or $n_{l'}^{(i)}$ can be recovered. Subsequently, the remaining erasures can be repaired sequentially.

\item[(III)] \emph{$x=2 :$} We consider two different cases.
\bit
\item[(a)] Assume that both erased edges belong to the same bipartite graph, say the $i^{th}$ bipartite graph. Note that any pair of edges in a bipartite graph can have at most one node in common. Hence, out of the $2L$ nodes $n_l^{(i)}, 1 \leq l \leq 2L$, there are atleast two nodes that are connected to only one of the two erased edges. Let those nodes be $n_l^{(i)}$ and $n_{l'}^{(i)}$, $l \neq l'$. Using similar arguments as in case (II), it can be proved that sequential recovery is possible in this case.
\item[(b)] Assume that the two erased edges belong to two different bipartite graphs. The edges are connected to $4$ distinct nodes. No more edges can be erased in this case. Hence at least one of the erased edges can be recovered unless the $4$ nodes are erased (which is not possible as it would make the total number of erasures 6). Once an erased edge is recovered, the remaining erasures can be repaired sequentially.
\eit
\item[(IV)] \emph{$x=3 :$} The bipartite graph considered here has a girth of at least $6$. Hence the code can recover from upto five (i.e., $(girth-1)$, see \cite{RawMazVis}) information symbol erasures sequentially, if none of the nodes $n_l^{(i)}, 1 \leq l \leq 2L, \ 1 \leq i \leq r$ are erased. Therefore, assume that a node $n_l^{(i)}$ is erased. But this node can be recovered using the parity check involving $N_l$ and the $r$ nodes connected to it. Subsequently, the $3$ erased information symbols can be recovered. 
\item[(V)] \emph{$x=4 :$} The bipartite graph considered here has a girth of at least $6$. Hence, even when $4$ edges are erased, there is at least one node that is connected to exactly one of the erased edges. Hence sequential recovery is possible. 
\een
\end{proof}

\section{Construction for Five Erasures}
The construction presented in this section generates an $(n,k,r,5)_\text{seq}$ code by adding a few parity symbols to the $(n,k,r,4)_\text{seq}$ code obtained from Construction~\ref{const:four_reg}. 
\begin{const} \label{const:five_regular}
Consider a code obtained from Construction~\ref{const:four_reg}. Let 
\bean
A_i & = & \{n_1^{(i)} \cdots n_L^{(i)}\},\\
B_i & = & \{n_{L+1}^{(i)} \cdots n_{2L}^{(i)}\}.
\eean
Consider the $2L$ nodes $N_1 \cdots N_{2L}$. Each of these $2L$ nodes store the parity of $r$ nodes in a fixed position in the $r$ copies of the bipartite graph. Out of these $2L$ nodes, $N_1 \cdots N_L$ represent the nodes connected to the nodes on one side of the bipartite graphs (see Fig.~\ref{fig:five_bipart}), i.e., nodes in $\cup_{i=1}^rA_i$. The set $\{N_1 \cdots N_L\}$ is partitioned into $(\lfloor \frac{L}{r}\rfloor)$ sets of size $r$ and one set of size $L \text{ mod } r$. Let the sets in this paritition be $C_1,...,C_{\lceil\frac{L}{r}\rceil}$. Add $\lceil\frac{L}{r}\rceil$ new nodes $S_1,...,S_{\lceil\frac{L}{r}\rceil}$ to the graph. The nodes in the set $C_i$ are connected to $S_i$. Each of the $\lceil\frac{L}{r}\rceil$ new nodes $S_i$ represent the parity of the nodes connected to it. All five erasure patterns that the $(n,k,r,4)_\text{seq}$ code cannot correct are of the form shown in Fig.~\ref{fig:five_bipart}. The new parity nodes added to the code can help in correcting each of these erasure patterns. Hence we obtain an $(n,k,r,5)_\text{seq}$ code having a rate of $\frac{Lr^2}{Lr^2+2Lr+2L+\lceil\frac{L}{r}\rceil}$.
\end{const}
\begin{figure}[H]
	\centering
	\includegraphics[width=6.5in]{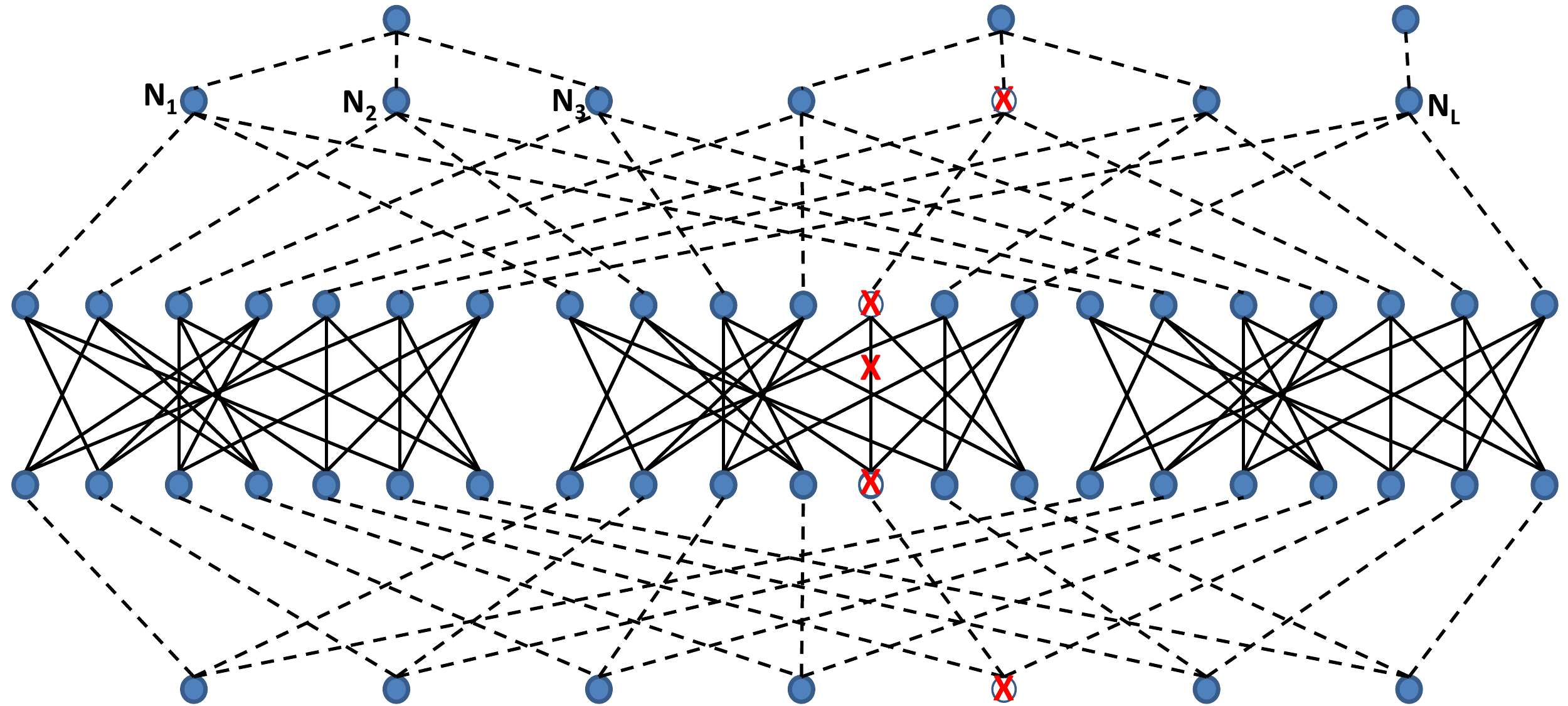}
	\caption[Example bipartite graph code for $5$ erasures]{An example code obtained from Construction~\ref{const:five_regular} for $r=3,L=7$. 'X' indicates a typical five erasure pattern that the $(n,k,r,4)_\text{seq}$ code in Construction~\ref{const:four_reg} cannot correct.}
	\label{fig:five_bipart}
\end{figure}

\bibliographystyle{IEEEtran}
\bibliography{bib_file}	

\end{document}